\documentclass[11pt,letterpaper]{article}

\usepackage{typearea}
\typearea{14}

\usepackage{comment}
\usepackage[ruled,linesnumbered,vlined]{algorithm2e}


\usepackage{color,colortbl}

\usepackage{float}
\usepackage{amsthm}
\usepackage{amsmath}
\usepackage{amssymb}
\usepackage{graphicx}
\usepackage{xspace}
\usepackage{nicefrac}
\usepackage[colorinlistoftodos,prependcaption,textsize=tiny]{todonotes}
\usepackage{caption}
\usepackage{url}

\usepackage{thmtools}
\usepackage{thm-restate}
\declaretheorem[name=Lemma]{lemma}

\usepackage{times,xcolor}
\usepackage{amsfonts,algorithmic}

\definecolor{Darkblue}{rgb}{0,0,0.4}
\definecolor{Brown}{cmyk}{0,0.61,1.,0.60}
\definecolor{Purple}{cmyk}{0.45,0.86,0,0}

\usepackage[colorlinks,linkcolor=Darkblue,filecolor=blue,citecolor=blue,urlcolor=Darkblue,pagebackref,pdfstartview=FitH]{hyperref}
\usepackage[nameinlink]{cleveref}

\newtheorem{theorem}{Theorem}
\newtheorem{corollary}{Corollary}

\newtheorem{observation}[lemma]{Observation}

\newtheorem{question}{Question}

\newcommand{\ddim}{{\rm ddim}}

\newcommand{\R}{\mathbb{R}}

\newcommand{\etal}{\emph{et. al.}}

\newcommand{\eps}{\epsilon}

\newcommand{\poly}{{\rm poly}}

\newcommand{\mommit}[1]{}
\newcommand{\namedref}[2]{\hyperref[#2]{#1~\ref*{#2}}}
\newcommand{\sectionref}[1]{\namedref{Section}{#1}}

\newcommand{\theoremref}[1]{\namedref{Theorem}{#1}}
\newcommand{\questionref}[1]{\namedref{Question}{#1}}

\newcommand{\figureref}[1]{\namedref{Figure}{#1}}
\newcommand{\algref}[1]{\namedref{Algorithm}{#1}}

\newcommand{\lemmaref}[1]{\namedref{Lemma}{#1}}

\newcommand{\observationref}[1]{\namedref{Observation}{#1}}
\newcommand{\corollaryref}[1]{\namedref{Corollary}{#1}}

\newcommand {\ignore} [1] {}

\usepackage{authblk}
\usepackage{pdfsync}

\begin{document}
\author[1]{Arnold Filtser}
\author[2]{Shay Solomon}
\affil[1]{Columbia University\thanks{The work was done while the author was affiliated with Ben-Gurion University of the Negev. Work supported in part by Simons Foundation, ISF grant 1817/17, and by BSF Grant 2015813.}. Email: \texttt{arnold273@gmail.com}}
\affil[2]{Tel Aviv University\thanks{Shay Solomon is partially supported by the Israel Science Foundation grant No.1991/19 and by Len Blavatnik and the Blavatnik Family foundation.
		Part of this work was carried out while this author was supported by the Rothschild Postdoctoral Fellowship, the Fulbright Postdoctoral Fellowship and the IBM Herman Goldstine Postdoctoral Fellowship.}. Email: \texttt{solo.shay@gmail.com}}

  \title{The Greedy Spanner is Existentially Optimal\footnote{A preliminary version of this paper appeared in PODC 2016 \cite{FS16}.}}
  \maketitle

\begin{abstract}
The greedy spanner is arguably the simplest and most well-studied spanner construction.
Experimental results demonstrate that it is at least as good as any other spanner construction, in terms of both the size and weight parameters.
However, a rigorous proof for this statement has remained elusive.

In this work we fill in the theoretical gap via a surprisingly simple observation: The greedy spanner is \emph{existentially optimal} (or existentially near-optimal)
for several important graph families, in terms of both the size and weight.
Roughly speaking, the greedy spanner is said to be existentially optimal (or near-optimal)
for a graph family $\mathcal G$ if the worst performance of the greedy spanner over all graphs in $\mathcal G$
is just as good (or nearly as good) as the worst performance of an optimal spanner over all graphs in $\mathcal G$.

Focusing on the weight parameter, the state-of-the-art spanner constructions for both general graphs
(due to Chechik and Wulff-Nilsen [SODA'16]) and doubling metrics (due to Gottlieb [FOCS'15]) are complex. Plugging our observation on these results, we conclude that the greedy spanner achieves near-optimal weight guarantees for both general graphs and doubling metrics, thus resolving two longstanding conjectures in the area.

Further, we observe that approximate-greedy spanners are existentially near-optimal as well.
Consequently, we provide an $O(n \log n)$-time construction of $(1+\epsilon)$-spanners for doubling metrics with constant lightness and degree.
Our construction improves Gottlieb's construction, whose runtime is $O(n \log^2 n)$ and whose number of edges and degree are unbounded,
and remarkably, it matches the state-of-the-art Euclidean result (due to  Gudmundsson et al.\ [SICOMP'02]) in all the involved parameters
(up to dependencies on $\epsilon$ and the dimension).
\end{abstract}
\section{Introduction}
\vspace{3pt}
{\bf 1.1~ Graph Spanners.~}
Given a (connected and undirected) $n$-vertex $m$-edge graph $G = (V,E,w)$ with positive edge weights and a parameter $t \ge 1$, a subgraph $H = (V,E',w)$ of $G$
($E' \subseteq E$) is called a \emph{$t$-spanner} for $G$ if for all $u,v \in V$, $\delta_H(u,v) \le t \cdot \delta_G(u,v)$.
(Here $\delta_G(u,v)$ and $\delta_H(u,v)$ denote the distances between $u$ and $v$ in the graphs $G$ and $H$, respectively.)
The parameter $t$ is called the \emph{stretch} of $H$.\footnote{More accurately, the \emph{stretch} of $H$ is the minimum number $t'$ such that $H$ is a $t'$-spanner for $G$, hence $t$ is in fact an upper bound on the stretch of $H$. However, referring to $t$ as the stretch of $H$ is a standard terminology in the area and is technically more convenient when the focus is existential bounds on spanner properties that depend on the stretch parameter, as in the current work.}
Spanners constitute a fundamental graph structure, and have been extensively and intensively studied since they were introduced \cite{PS89,PU89}.

In many practical applications one is required to construct a spanner that satisfies a number of useful properties, while preserving a small stretch.
First, the spanner $H$ should have a small number of edges.
Second, its \emph{weight} $w(H) = \sum_{e \in E} w(E)$ should be close to the weight of a minimum spanning tree (MST) of the graph $G$.
We henceforth refer to the normalized notion of weight $\Psi(H) = \frac{w(H)}{w(MST(G))}$, which is called \emph{lightness};
a \emph{light} spanner is one with small lightness.
Third, its \emph{degree} $\Delta(H)$, defined as the maximum number of edges incident on a vertex, should be small.

Light and sparse spanners are particularly useful for efficient broadcast protocols in the message-passing model of distributed computing \cite{ABP90,ABP91},
where efficiency is measured with respect to both the total communication cost (corresponding to the spanner's size and weight) and the speed of message delivery at all destinations (corresponding to the spanner's stretch).
Additional applications of such spanners in distributed systems include network synchronization and computing global functions \cite{Awerbuch85,PU89,ABP90,ABP91,Peleg00}. 
Light and sparse spanners were also found useful for various data gathering and dissemination tasks in overlay networks \cite{BKRCV02,VWFME03,KV01}, 
in wireless and sensor networks \cite{RW04,BDS04,SS10}, 
for VLSI circuit design \cite{CKRSW91,CKRSW292,CKRSW92,SCRS01},
for routing \cite{WCT02,PU89,PU89b,TZ01}, 
to compute distance oracles and labels \cite{Peleg00Prox,TZ01b,RTZ05}, 
and to compute almost shortest paths \cite{Coh98,RZ11,Elkin05,EZ06,FKMSZ05}. 
Low degree spanners are also very useful in many of these applications.
For example, the degree of the spanner is what determines local memory constraints when using
spanners to construct network synchronizers  and efficient broadcast protocols.
In compact routing schemes, the use of low degree spanners enables the routing tables to be of small size.
More generally, viewing vertices as processors, in many applications the degree of a processor
represents its \emph{load}, hence a low degree spanner guarantees that the load on all the processors in the network will be low.

The \emph{greedy spanner} by Alth$\ddot{\mbox{o}}$fer  et al.\ \cite{ADDJS93} is arguably the simplest and most well-studied spanner construction.    			
Alth$\ddot{\mbox{o}}$fer  et al.\ showed that for every  weighted $n$-vertex graph $G = (V,E,w)$ and an integer parameter $k \ge 1$,
the \emph{greedy algorithm} (see \algref{fig:greedy_alg}) constructs a $(2k-1)$-spanner with $O(n^{1+1/k})$ edges;
assuming Erd\H{o}s' girth conjecture \cite{erd64}, this size bound is asymptotically tight.
Alth$\ddot{\mbox{o}}$fer  et al.\ also showed that the lightness of the greedy spanner is $O(n/k)$.
Chandra et al.\ \cite{CDNS92} improved the lightness bound, and showed that the greedy spanner for stretch parameter $t=(2k-1)\cdot(1+\epsilon)$
(here $k>1$, $\epsilon>0$) has lightness $O(k \cdot n^{1/{k}} \cdot (1/\epsilon)^{1+1/k})$.
Two decades later, Elkin, Neiman and the second author \cite{ENS14} improved the analysis of \cite{CDNS92} and showed that the greedy $(2k-1)\cdot(1+\epsilon)$-spanner has lightness $O(n^{1/k}\cdot (1+ k/(\epsilon^{1+1/k}\log k)))$.
In a very recent breakthrough, Chechik and Wulff-Nilsen \cite{CW18} improved the lightness bound all the way to $O(n^{1/k} (1/\eps)^{3+2/k})$.
Assuming Erd\H{o}s' girth conjecture \cite{erd64} and ignoring dependencies on $\epsilon$, the bound of \cite{CW18} on the lightness is asymptotically tight,
thus resolving a major open question in this area. However, the result of Chechik and Wulff-Nilsen \cite{CW18}
is not due to a refined analysis of the greedy spanner. Instead, they devised a different construction, which is far more complex, and
bounded the lightness of their own construction. The following question was left open.
\begin{question} \label{question1}
	Is the lightness analysis of \cite{ENS14} for the greedy spanner optimal, or can one refine it to derive a stronger bound?
	In particular, is the spanner of \cite{CW18} lighter than the greedy spanner?
\end{question}
\vspace{2pt}
\noindent{\bf 1.2~ Spanners for Euclidean and Doubling Metrics.~}
Consider a set $P$ of $n$ points in $\mathbb R^d$, $d \ge 2$, and a stretch parameter $t \ge 1$.
A graph $G = (P,E,w)$ in which the weight $w(p,q)$ of each edge $e = (p,q) \in E$
is equal to the Euclidean distance $\|p-q\|$ between $p$ and $q$ is called a \emph{Euclidean graph}.
We say that the Euclidean graph $G$ is a \emph{$t$-spanner} for $P$ (or equivalently, for the corresponding Euclidean metric $(P,\|\cdot\|)$) if for every pair $p,q \in P$ of distinct points, there exists a path $\Pi(p,q)$ in $G$
between $p$ and $q$ whose weight (i.e., the sum of all edge weights in it) is at most $t \cdot \|p-q\|$.
The path $\Pi(p,q)$ is said to be a \emph{$t$-spanner path} between $p$ and $q$.
For Euclidean metrics, one usually focuses on the regime $t = 1+\eps$, for $\eps > 0$ being an arbitrarily small parameter.
Euclidean spanners were introduced by Chew \cite{Chew86}, and were subject to intensive ongoing research efforts since then.
We refer to the book ``Geometric Spanner Networks'' \cite{NS07}, which is devoted almost exclusively to Euclidean spanners and their numerous applications.
As with general graphs, it is important to devise Euclidean spanners that achieve small size, lightness and degree.

The \emph{doubling dimension} of a metric space $(M,\delta)$ is the smallest value $\ddim$
such that every ball $B$ in the metric space can be covered by at most
$2^{\ddim}$ balls of half the radius of $B$.
This notion generalizes the Euclidean dimension, since the doubling dimension
of the Euclidean space $\mathbb R^d$ is $\Theta(d)$.
A metric space is called \emph{doubling} if its doubling dimension is constant.
Spanners for doubling metrics were also subject of intensive research  \cite{GGN04,CGMZ16,CG09,HM06,Roditty12,GR081,GR082,Smid09,ES15,Sol14}.
The basic line of work in this context is to generalize the known Euclidean spanner results for arbitrary doubling metrics.

Das et al.\ \cite{DHN93} showed that, in low-dimensional Euclidean metrics $\R^d$, the greedy $(1+\eps)$-spanner has constant degree (and so $O(n)$ edges)
and $O(1/\eps)^{2d})$ lightness.
In $n$-point doubling metrics, the greedy $(1+\eps)$-spanner has $O(n)$ edges and lightness $O(\log n)$ \cite{Smid09}.
As for the degree, there exist $n$-point metric spaces with doubling dimension 1 for which the greedy spanner has a degree of $n-1$ \cite{HM06,Smid09}.
It has been a major open question to determine whether any doubling metric admits
a $(1+\eps)$-spanner  with sub-logarithmic lightness.
A breakthrough paper of Gottlieb \cite{Got15} answered this fundamental question in the affirmative by devising such a spanner construction with constant lightness.
Again, this result is not due to a refined analysis of the greedy spanner. Instead, Gottlieb devised a different construction,
which is far more complex, and bounded the lightness of his own construction. The following question was left open.
\begin{question} \label{question2}
	Is the lightness analysis of \cite{Smid09} for the greedy spanner optimal, or can one refine it to derive a stronger bound?
	In particular, is the spanner of \cite{Got15} lighter than the greedy spanner?
\end{question}

The relatively high runtime of the greedy spanner is a drawback.
The state-of-the-art implementation of the greedy spanner in both Euclidean and doubling metrics requires time $O(n^2 \log n)$ \cite{BCFMS10} (although there are some heuristics that might be useful in practice \cite{ABBB15,ABBB17}).
Building on \cite{DHN93}, Das and Narasimhan \cite{DN97} devised a much faster algorithm that follows the greedy approach.
The runtime of their ``approximate-greedy'' algorithm is $O(n \log^2 n)$, yet its degree and lightness are both bounded by constants (as with the greedy spanner).
Gudmundsson et al.\ \cite{GLN02} improved the result of \cite{DN97}, implementing the approximate-greedy algorithm within time $O(n \log n)$.
For doubling metrics, however, the only spanner construction with sub-logarithmic lightness is that of   \cite{Got15};
the runtime of Gottlieb's construction is $O(n \log^2 n)$ rather than $O(n \log n)$, and the size and degree of his construction are unbounded.
Hence, there is a big gap in this context between Euclidean and doubling metrics, leading to the following question.
\begin{question} \label{question3}
	Can one compute $(1+\eps)$-spanners with constant lightness in doubling metrics  within time $O(n \log n)$?
	Furthermore, can one extend the state-of-the-art Euclidean result of \cite{GLN02} to arbitrary doubling metrics?
\end{question}

There have been numerous experimental studies on Euclidean spanners. (See \cite{FG05,Far08}, and the references therein.)
The conclusion emerging from these experiments is that the greedy Euclidean spanner outperforms the other popular Euclidean spanner constructions, with respect to the
size and lightness bounds. (Specifically, the greedy spanner was found to be $10$ times sparser and $30$ times lighter than any other examined spanner.)
It is reasonable to assume that a similar situation occurs in arbitrary doubling metrics.
\vspace{6pt}
\\
\noindent{\bf 1.3~  Our Contribution.~}
In this work we fill in the theoretical gap by making three important observations.
\begin{enumerate}
	\item Our first observation is surprisingly simple: The greedy spanner is \emph{existentially optimal} with respect to both the size and the lightness,
	for any graph family  that is closed under edge removal.
	
	\vspace{5pt}	
	Applying this observation to the family of general weighted graphs, we conclude that the greedy spanner is just as light as the spanner of \cite{CW18}, thus answering \questionref{question1}.
	
	\vspace{5pt}
	Moreover, it is known that the greedy spanner 
	can be easily implemented within time $O(m(n^{1 + 1/k} + n\log n))$ (cf.\ \cite{ES16}),
	and is thus much faster than the complex algorithm of  \cite{CW18}.
	(Although the runtime of the algorithm of \cite{CW18} is not analyzed explicitly therein,
	a naive implementation of that algorithm,
	which involves diameter computations of carefully selected subgraphs, incurs a runtime of $\Omega(m^2 n)$.)
	We remark that all faster spanner constructions (e.g. \cite{BS07,ES16,MPVX15,EN17,ADFSW19}) achieve a worse lightness bound than that of the greedy spanner.
	Consequently, the greedy algorithm enjoys the fastest known runtime of any $(2k-1)(1+\eps)$-spanner with $O(n^{1+1/k})$ edges and  lightness $O(n^{1/k} (1/\eps)^{3+2/k})$,
	or in other words, it is the fastest algorithm for constructing spanners that are near-optimal with respect to all the involved parameters (stretch, size and lightness).
	\vspace{6pt}
	\item The first observation does not hold for doubling metrics.
	Our second observation is that the greedy spanner is existentially \emph{near}-optimal with respect to both the size and the lightness,
	for the family of doubling metrics.
	In particular, it is just as light as the spanner of \cite{Got15}, thus answering \questionref{question2}.
	\vspace{6pt}
	\item Our third observation concerns the optimality of the approximate-greedy algorithm of \cite{DN97,GLN02} in doubling metrics, and is more intricate than the first two observations.
	Informally, it states that the approximate-greedy spanner with stretch parameter $t$ is existentially near-optimal with respect to the lightness,
	for the family of doubling metrics, but when compared to spanners with a slightly smaller stretch parameter $t' < t$.
	This enables us to conclude that the   lightness of the approximate-greedy spanner is close to that of \cite{Got15}.
	In this way we manage to extend the state-of-the-art Euclidean result of \cite{GLN02} to arbitrary doubling metrics, thus answering
	\questionref{question3}.\footnote{The $O(n \log n)$ runtime bound of \cite{GLN02} holds in the traditional algebraic computation-tree model with
		the added power of indirect addressing. Our result applies with respect to the same computation model.}
	
	\vspace{5pt}	
	While our approximate-greedy spanner achieves the same lightness bound as that of Gottlieb's spanner \cite{Got15}, it has several important advantages over it.
	First, our algorithm is conceptually much simpler than that of \cite{Got15}.
	Second, our spanner construction has constant degree (ignoring dependencies on $\eps$ and on the doubling dimension), while the degree in \cite{Got15} is not analyzed, and naively it may be as large as $\Omega(n)$.
	Third, the degree bound of our spanner implies that it has $O(n)$ edges, while the size in \cite{Got15} is not analyzed,
	and it may be significantly larger than the size of our spanner.
	Finally, the construction time of \cite{Got15} is $O(n \log^2 n)$, while our construction time is $O(n \log n)$, which is considered the holy grail in the area of Euclidean and doubling spanners.
	
	\vspace{5pt}	
	Our third observation concerning the existential near-optimality of approximate-greedy spanner algorithms
	can be viewed as a \emph{general paradigm} for obtaining \emph{fast} spanner constructions that are both sparse and light.
	Although we apply it in this paper only to the family of doubling metrics, the same paradigm can be applied to other families of graphs, including general graphs, by adjusting it appropriately. In fact, the very recent work of Alstrup et. al. \cite{ADFSW19} follows this paradigm to obtain, for general graphs,
	either a $(2k-1)(1+\eps)$-spanner with $O(n^{1+1/k}\cdot\poly(\frac1\eps))$ edges and $O(n^{1/k}\cdot\poly(\frac1\eps))$ lightness
	in $O(n^{2+\zeta+1/k})$ time
	or an $O(k)$-spanner with the same size and lightness bounds in $O(n^{1+\zeta+1/k})$ time, where $\zeta > 0$ is a small constant.
\end{enumerate}

To summarize, by introducing and studying a new notion of optimality, \emph{existential (near-)optimality},
this paper provides an extremely simple yet powerful tool in the area of spanners.
We believe that the notion of existential optimality, defined formally below, is of fundamental importance, and we anticipate that it will find more applications, even outside the area of spanners.
\vspace{7pt}
\\
\noindent
{\bf Some definitions concerning existential optimality.~}
Although the meaning of \emph{existential optimality} can be   understood from the context, it is instructive to provide a formal definition.
Fix an arbitrary stretch parameter $t \ge 1$ and some graph family $\mathcal{G}$.
For a graph $G \in \mathcal{G}$, let $OPT^{sparse}_t(G)$ (respectively, $OPT^{light}_t(G)$)
denote the optimal size (resp., lightness) of any $t$-spanner for $G$,
and let $OPT^{sparse}_t(\mathcal{G}) = \max\{OPT^{sparse}_t(G) ~\vert~ G \in \mathcal{G}\}$
(resp., $OPT^{light}_t(\mathcal{G}) = \max\{OPT^{light}_t(G) ~\vert~ G \in \mathcal{G}\}$)
denote the maximum value $OPT^{sparse}_t(G)$ (resp., $OPT^{light}_t(G)$) over all graphs $G$ in $\mathcal{G}$.
The greedy $t$-spanner is said to be \emph{existentially optimal} with respect to the size (respectively, lightness) if
for any graph $G \in \mathcal{G}$, the size (resp., lightness) of the greedy $t$-spanner for $G$ does not exceed $OPT^{sparse}_t(\mathcal{G})$  (resp., $OPT^{light}_t(\mathcal{G})$).
This does not mean that the size (respectively, lightness) of the greedy  $t$-spanner
for any graph $G \in \mathcal{G}$ is bounded by $OPT^{sparse}_t(G)$ (resp., $OPT^{light}_t(G)$).
It simply means that there \emph{exists} a graph $G' \in \mathcal{G}$,
such that the size (resp., lightness) of the greedy $t$-spanner for $G$ is bounded by $OPT^{sparse}_t(G')$ (resp., $OPT^{light}_t(G')$).
In other words, the maximum size (resp., lightness) of the greedy $t$-spanner over all graphs in $\mathcal G$
is equal to the maximum size (resp., lightness) of an optimal $t$-spanner over all graphs in $\mathcal G$.

{For example, let $\mathcal{G}$ be the family of general weighted graphs on $n$ vertices, and let $H$ be an $n$-vertex dense graph of high girth,
	namely, with girth $t+2$ and $n^{1+\Theta(1/t)}$ edges, where all edge weights are 1.
	Also, let $S$ be a star on the same vertex set as $H$ rooted at an arbitrary vertex,
	so that all edges of $S$ that belong to $H$ have weight 1 and all edges of $S$ that do not belong to $H$ have weight $1+\eps$.
	Finally, let $G$ be the graph containing all edges of $H$ and all edges of $S$ with weight $1+\eps$.
	Note that the greedy $t$-spanner for $G$ includes all $n^{1+\Theta(1/t)}$ edges of the high girth graph $H$,
	whereas the optimal $t$-spanner (assuming $t \ge 2+2\eps$) consists of the edges of the star $S$, hence is much sparser and lighter.
	(See \figureref{fig:girth} for an illustration.)
	This example, however, does not contradict the existential optimality of the greedy spanner: Although
	the size (respectively, lightness) of the greedy $t$-spanner for $G$ exceeds $OPT^{sparse}_t(G)$ (resp., $OPT^{light}_t(G)$), it can be shown that it is equal to
	$OPT^{sparse}_t(H)$ (resp., $OPT^{light}_t(H)$), which, in turn, is bounded by $OPT^{sparse}_t(\mathcal{G})$  (resp., $OPT^{light}_t(\mathcal{G})$).}
\begin{figure}
	\begin{center}
		\includegraphics[width=0.41\textwidth]{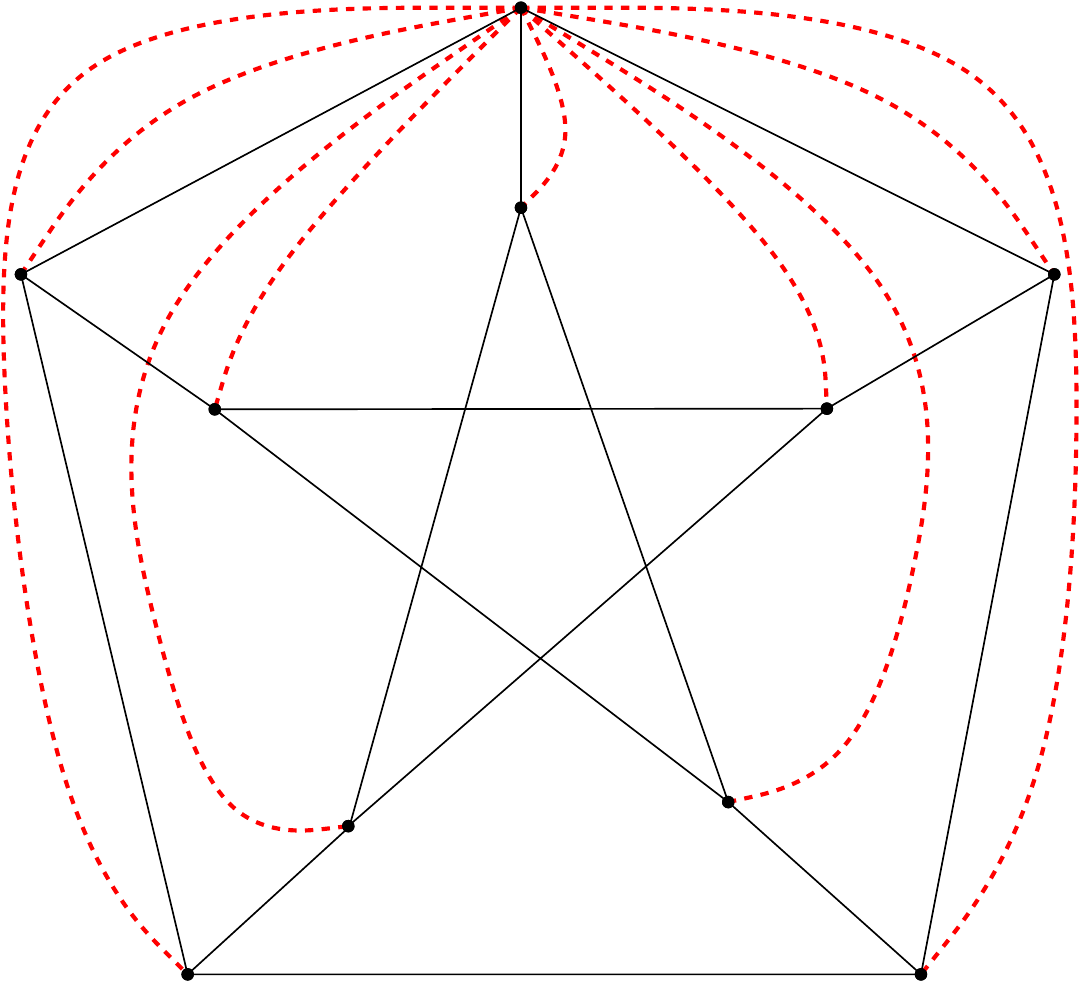}
		\caption{\small  The graph $H$ in the figure is the Petersen graph on 10 vertices, with girth 5 and 15 edges.
			All edges of $H$ have weight 1, and are colored black.
			The red dashed edges are the edges of the star $S$ of weight $1+\eps$.
			The graph $G$ is obtained as the union of the black and red edges in the figure.
			The greedy 3-spanner for $G$ includes all 15 edges of $H$, whereas the optimal 3-spanner for $G$ consists of the 9 edges of $S$.}				
		\label{fig:girth}
	\end{center}
\end{figure}

The meaning of \emph{existential near-optimality} is similar, except that we are allowed to have some \emph{slack},
which may depend on the stretch parameter $t$ as well as on parameters of the graph family of interest $\mathcal{G}$.
As mentioned, in our third observation we compare the lightness of the greedy spanner with a certain stretch parameter $t$
to the optimal lightness of any spanner, but with a slightly smaller stretch parameter $t'$.
This is just one example of how the slack parameter can be used.
Another example is to compare the greedy spanner in some graph family $\mathcal{G}$ to an optimal spanner, but with respect to a different (closely related) graph family $\mathcal{F'}$.
In particular, in our second and third observations we compare the lightness of the greedy spanner in metric spaces of bounded doubling dimension
to the optimal lightness of any spanner, but with respect to metric spaces of slightly larger doubling dimension.
It would be interesting to study additional ways of using the slack parameter,
as they may lead to new results in this area.

We remark that light spanners were extensively studied in various graph families such as planar graphs \cite{ADDJS93,Klein05}, apex graphs \cite{GS02},
bounded pathwidth graphs \cite{GH12}, bounded genus graphs \cite{Grigni00,GS02,DHM10}, bounded treewidth graphs \cite{DHM10},
and graphs excluding fixed minors \cite{Grigni00,DHM10,BLW17Minor}.
Since all these graph families are closed under edge removal, our first observation implies that the greedy spanner for them is just as good as any other spanner.
\vspace{6pt}
\\
\noindent{\bf 1.4~  Subsequent work.~}
The preliminary version of this paper appeared in PODC 2016 \cite{FS16}, and it triggered further work in the area.
We focus here on the two most relevant follow-up papers \cite{BLW19,LS19}. 

Borradaile, Le and Wulff{-}Nilsen \cite{BLW19}, improving over  \cite{Got15}, presented a construction of $(1+\eps)$-spanners for doubling metrics with lightness 
$(1/\eps)^{O(\ddim)}$; the improvement is in the dependence of the lightness bound on $\eps$ and $\ddim$ (see Section \ref{sebsec:DoublingPreliminaries} for the details).
Moreover, by building on our Theorem \ref{finish}, Borradaile et al.\  achieved an $O(n \log n)$-time algorithm for constructing $(1+\eps)$-spanners with lightness $(1/\eps)^{O(\ddim)}$.

Le and the second author \cite{LS19} studied the greedy spanner in $d$-dimensional Euclidean metrics and determined the exact asymptotic dependencies on $\eps$ and $d$ for 
both the size and lightness bounds for any $d = O(1)$ (disregarding polylogarithmic factors of $1/\eps$): $\Theta(n(1/\eps)^{d-1})$ edges and lightness $\Theta((1/\eps)^{d})$.
Moreover, Le and Solomon \cite{LS19} showed that Steiner points lead to a quadratic improvement in the size of Euclidean spanners.
\vspace{6pt}
\\
\noindent{\bf 1.5~  Organization.~}
In \sectionref{sec:pre} we present the notation that is used throughout the paper, and summarize some statements from previous work that are most relevant to us.
In \sectionref{sec:Greedy_optimal} we show that the greedy spanner is existentially optimal for graph families that are closed under edge removal.
The basic optimality argument of \sectionref{sec:Greedy_optimal} is extended to doubling metrics in \sectionref{sec:doubling}.
Finally, in \sectionref{sec:Fast_alg} we show that the approximate-greedy spanner in doubling metrics is light.

\section{Preliminaries}\label{sec:pre}
All the graphs considered in the paper are connected, undirected and with positive edge weights. Let $G=(V,E,w)$ be such a graph.
The weight $w(P)$ of a path $P$ is the sum of all edge weights in it, i.e., $w(P) = \sum_{e \in P} w(e)$.
For a pair of vertices $u,v \in V$, let $\delta_G(u,v)$ denote the distance between $u$ and $v$ in $G$, i.e., the weight of a shortest path between them.
We denote by $M_G = (V,\delta_G)$ the (shortest path) metric space \emph{induced} by $G$; we will view $M_G$ as a complete weighted graph $(V, {V \choose 2},w)$ over the vertex set $V$,
where the weight $w(u,v)$ of an edge $(u,v)$ is given by the graph distance $\delta_G(u,v)$ between its endpoints.
A subgraph $H=(V,E',w)$ of $G$ (where $E'\subseteq E$) is called a \emph{t-spanner} for $G$ if for all $u,v \in V$,  $\delta_H(u,v)\le t\cdot \delta_G(u,v)$.
The parameter $t$ is called the \emph{stretch} of the spanner $H$.
If $\delta_H(u,v)\le t\cdot \delta_G(u,v)$ for all edges $(u,v) \in E$, then
it also holds that $\delta_H(u,v)\le t\cdot \delta_G(u,v)$ for all pairs of vertices $u,v \in V$.
Therefore, to bound the stretch of the spanner, one may restrict the attention to the edges of the graph.
Let $|H| = |E'|$ denote the \emph{size} of $H$, and let $w(H) = w(E') = \sum_{e\in E'}w(e)$ denote its \emph{weight}.
The \emph{lightness} $\Psi(H)$ of $H$ is the ratio between the weight of $H$ and the weight of an MST for $G$, i.e., $\Psi(H)=\frac{w(H)}{w(MST(G))}$.
(Throughout the paper all logarithms are in base 2.)

We refer the reader to Section 1.3 for some definitions concerning the notion of existential optimality.

The result of  Chechik and Wulff{-}Nilsen \cite{CW18} is summarized in the following theorem.
\begin{theorem}[\cite{CW18}]\label{thm:CW18}
	For every weighted $n$-vertex graph $G=(V,E,w)$ and parameters $k\ge 1$ and $0<\epsilon<1$, there exists a
	$(2k-1)\cdot(1+\epsilon)$-spanner with $O(n^{1+1/k})$ edges and lightness $O(n^{1/k} (1/\eps)^{3+2/k})$. Such a spanner can be constructed in polynomial time.
\end{theorem}	

\subsection{Doubling metrics}\label{sebsec:DoublingPreliminaries}
Most statements in this section will be used in our construction and analysis of spanners for doubling metrics. 

We start with the
standard packing property in doubling metrics (see, e.g.,  \cite{GKL03}).
\begin{lemma} \label{lem:doubling_packing}
	Let $(M,\delta)$ be a metric space  with doubling dimension $\ddim$.
	If $S \subseteq M$ is a subset of points with minimum interpoint distance $r$
	that is contained in a ball of radius $R$, then
	$|S| \le \left(\frac{2R}{r}\right)^{O(\ddim)}$.
\end{lemma}

The following theorem states that any doubling metric admits a constant degree $(1+\eps)$-spanner. This theorem will be useful in answering Question \ref{question3}, and more specifically for achieving the degree bound required for extending the state-of-the-art Euclidean result of \cite{GLN02} to arbitrary doubling metrics (see \theoremref{finish}).

\begin{theorem}[\cite{CGMZ16,GR08}]\label{thm:doubling_degree}
	For any $n$-point metric  space $(M,\delta)$ with doubling dimension
	$\ddim$ and parameter $0<\epsilon<1/2$, there exists a $(1 + \epsilon)$-spanner with degree $(1/\epsilon)^{O(\ddim)}$.
	The runtime of this construction is $(1/\eps)^{O(\ddim)}(n \log n)$.
\end{theorem}

The result of Smid \cite{Smid09} is summarized in the following theorem. We provide this theorem only for the sake of comparison with our new result (see \corollaryref{cor:Greedy_G}),
which improves the logarithmic lightness bound provided by this theorem to constant.
\begin{theorem}[\cite{Smid09}]\label{thm:smid09}
	For any $n$-point metric space $(M,\delta)$ with doubling dimension $\ddim$ and any parameter $0<\epsilon<\frac{1}{2}$,  the greedy $(1+\epsilon)$-spanner has
	$(1/\epsilon)^{O(\ddim)}n$  edges and lightness $(1/\epsilon)^{O(\ddim)}\log n$.
\end{theorem}

The result of Gottlieb \cite{Got15} is summarized in the following theorem. We will use this theorem for answering Questions \ref{question2} and \ref{question3}.
\begin{theorem}[\cite{Got15}]\label{thm:Got15}
	For any $n$-point metric space $(M,\delta)$ with doubling dimension $\ddim$ and parameter $0<\epsilon<1/2$,
	there exists a $(1+\epsilon)$-spanner  with lightness $(\ddim/\epsilon)^{O(\ddim)}$.
	The runtime of this construction is $(\ddim/\epsilon)^{O(\ddim)} (n\log^{2}n)$.
\end{theorem}
{\bf Remark.} Recently, Borradaile, Le and Wulff{-}Nilsen \cite{BLW19} showed that the greedy $(1+\eps)$-spanner has lightness  $(1/\epsilon)^{O(\ddim)}$ 
in doubling metrics, improving over the lightness bound provided by Theorem \ref{thm:Got15}. 

\subsection{The Greedy Spanner and its Basic Properties}
\begin{algorithm}
	\caption{$\texttt{Greedy}(G=(V,E,w),t)$}\label{fig:greedy_alg}
	\begin{algorithmic}[1]
		\STATE $H=(V,\emptyset,w)$.
		\FOR {each edge $(u,v)\in E$, in non-decreasing order of weight,}
		\IF {$\delta_H(u,v)>t\cdot w(u,v)$}
		\STATE Add the edge $(u,v)$ to $E(H)$.
		\ENDIF
		\ENDFOR
	\end{algorithmic}
\end{algorithm}
The greedy spanner algorithm is presented in \algref{fig:greedy_alg}.
Let $H=\left(V,E_{H},w\right)$ be the output of an arbitrary execution of the greedy algorithm with stretch parameter $t$.
It is immediate that $H$ has stretch at most $t$.
If the edge weights in the graph are distinct, then $H$ is uniquely defined, but this does not hold in general;
nevertheless, by letting $H$ designate an arbitrary such spanner, we may henceforth refer to it as the \emph{greedy $t$-spanner}.
The following observation is immediate (see, .e.g., \cite{ENS14,CW18}).
\begin{observation}
	\label{fct:greedy contains MST}
	$H$ contains all   edges of some MST of $G$, denoted $Z$. (Hence $Z$ is also an MST of $H$.)
\end{observation}

\section{The Basic Optimality Proof}\label{sec:Greedy_optimal}
In this section we show that the greedy spanner is existentially optimal, with respect to both the size and the lightness, for any graph family that is closed under edge removal.

We start by making the basic observation that the only $t$-spanner of the greedy $t$-spanner is itself.
\begin{lemma}
	\label{obs:spanner_of_greedy}
	Let $G=(V,E,w)$ be any weighted graph, let $t \ge 1$ be any stretch parameter, and let $H$ be the greedy $t$-spanner of $G$.
	If $H'$ is a $t$-spanner for $H$, then $H' = H$.
\end{lemma}
\begin{proof}
	Assume for contradiction that $H'$ is a $t$-spanner for $H$ yet there is an edge $e \in H \setminus H'$.
	Let $P$ be a shortest path in $H'$ between the endpoints of $e$. As $H'$ is a $t$-spanner of $H$,
	it holds that $w(P)\le t\cdot w(e)$.
	Consider the last edge examined by the greedy algorithm among the edges of $P$ and $e$, denoted $e'$.
	By the description of the greedy algorithm, we have $w(e) \le w(e')$.
	Consequently, by the time the greedy algorithm examines edge $e'$, all the edges of the path $(P \cup e) \setminus e'$ must have already been added to the greedy spanner.
	(See \figureref{fig:lem3} for an illustration.)
	This path connects the endpoints of $e'$, and its weight is given by $$w(P) - w(e') + w(e) ~\le~ w(P) ~\le~ t \cdot w(e) ~\le~ t \cdot w(e').$$
	Hence the greedy algorithm will not add edge $e'$ to $H$, a contradiction.
	\begin{figure}
		\begin{center}
			\includegraphics[width=0.41\textwidth]{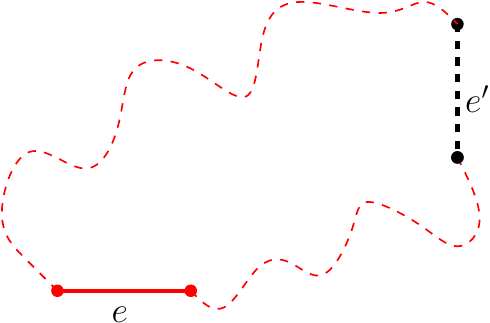}
			\caption{\small The path $P$ in $H'$ between the endpoints of edge $e$ is depicted by a dashed line.
				The path $P\cup e \setminus e'$ between the endpoints of edge $e'$, all edges of which have been added to $H$ by the time the greedy algorithm examines edge $e'$,
				is colored red.}
			\label{fig:lem3}
		\end{center}
	\end{figure}
\end{proof}

Equipped with \lemmaref{obs:spanner_of_greedy}, we now turn to the basic optimality proof.
\begin{theorem}[Greedy is existentially optimal]\label{thm:main}
	Let $\mathcal{G}$ be any family of $n$-vertex graphs that is closed under edge
	removal, and let $t = t(n) \ge 1$ be any stretch parameter.
	For every graph $G\in\mathcal{G}$,	the greedy $t$-spanner $H$ of $G$ has at most $OPT^{sparse}_t(\mathcal{G})$ edges and lightness at most $OPT^{light}_t(\mathcal{G})$.
	In other words, the greedy $t$-spanner is existentially optimal for $\mathcal{G}$ with respect to both the size and the lightness.
\end{theorem}
\begin{proof}
	Consider an arbitrary graph $G$ in $\mathcal{G}$, and let $H$ be the greedy $t$-spanner of $G$.
	Since $\mathcal{G}$ is closed under edge removal and $H$ is a subgraph of $G$,  $H$ belongs to $\mathcal{G}$.
	Hence, there exist $t$-spanners $\mathcal{H}^{sparse}$ and $\mathcal{H}^{light}$ of $H$ with at most $OPT^{sparse}_t(\mathcal{G})$  edges and lightness at most $OPT^{light}_t(\mathcal{G})$, respectively.
	\lemmaref{obs:spanner_of_greedy} implies that $\mathcal{H}^{sparse}=\mathcal{H}^{light}=H$, from which the size bound on $H$ immediately follows. 
	The lightness bound is slightly trickier, as the spanner $\mathcal{H}^{light}$ is computed on top of the greedy spanner $H$ rather than the original graph $G$.
	Nevertheless, \observationref{fct:greedy contains MST} implies that $G$ and $H$ have the same MST $Z$.
	Since the lightness of $\mathcal{H}^{light}$ is at most $OPT^{light}_t(\mathcal{G})$ and $Z$ is an MST for $H$, it follows that
	$$OPT^{light}_t(\mathcal{G}) ~\ge~ \Psi(\mathcal{H}^{light}) ~=~ \frac{w(\mathcal{H}^{light})}{w(MST(H))} ~=~ \frac{w(\mathcal{H}^{light})}{w(Z)}~.$$
	Using the fact that $\mathcal{H}^{light}= H$, we conclude that the lightness of $H$ satisfies
	$$\Psi(H) ~=~ \frac{w(H)}{w(MST(G))}  ~=~ \frac{w(H)}{w(Z)} ~=~ \frac{w(\mathcal{H}^{light})}{w(Z)} ~\le~
	OPT^{light}_t(\mathcal{G}).$$
\end{proof}

As the family of weighted graphs is closed under edge removal, we can apply \theoremref{thm:main} on it.
Hence the greedy spanner for general graphs has size and lightness at least as good as in \theoremref{thm:CW18}.
\begin{corollary}\label{cor:Greedy_CW}
	For every weighted graph $G=(V,E,w)$ on $n$ vertices and $m$ edges and parameters $k\ge 1$ and $0<\epsilon<1$, the greedy
	$(2k-1)\cdot(1+\epsilon)$-spanner has $O(n^{1+1/k})$ edges and lightness $O(n^{1/k} (1/\eps)^{3+2/k})$.
	(A naive implementation of the greedy algorithm requires $O(m n^{1+1/k})$ time.)
\end{corollary}	

In \cite{BFN19} it was proved that for any parameter $0<\delta<1$ and any stretch parameter $t = t(n)$,
if every  $n$-vertex weighted graph  admits a $t$-spanner with at most $m(n,t)$ edges and lightness  at most $l(n,t)$, then for every such graph
there also exists a $t/\delta$-spanner with at most $m(n,t)$ edges and lightness at most  $1+\delta\cdot l(n,t)$.
Applying \theoremref{thm:main} again, we derive the following result.
\begin{corollary}\label{cor:light_Greedy}
	For every weighted $n$-vertex graph $G = (V,E,w)$ and parameter $0<\delta<1$,
	the greedy $O(\log n/\delta)$-spanner has $O(n)$ edges and lightness at most $1+\delta$.
\end{corollary}	

As mentioned in the introduction, a plethora of graph families that are closed under edge removal were studied extensively in the spanner literature.
This includes the families of planar graphs, bounded genus graphs, bounded treewidth graphs, graphs excluding fixed minors, and more.
For all these graph families, \theoremref{thm:main} shows that the greedy spanner is existentially optimal.

\section{The Optimality Argument in Doubling Metrics}\label{sec:doubling}
The basic optimality argument of \sectionref{sec:Greedy_optimal} applies to graph families that are closed under edge removal.
Note that metric spaces do not fall into this category.
Nevertheless, for metric spaces, the basic optimality argument suffices: On the one hand, the upper bound for general weighted graphs applies to any metric space,
and on the other hand, the lower bound due to high girth graphs naturally applies to the induced metric spaces (see, e.g., \cite{ADDJS93,RR98}).

In this section we study the optimality of the greedy spanner for \emph{doubling metrics}.
For such metric spaces, one would like to obtain spanners with stretch $1+\eps$, where $\eps$ is arbitrarily close to 0.
We will show that the greedy $(1+\eps)$-spanner is existentially near-optimal in doubling metrics, with respect to both the size and the lightness.
The next observation and subsequent lemma will be used for proving the lightness optimality.
\begin{observation} \label{mst-metric}
	Consider the metric space $M_G$ induced by an arbitrary weighted graph $G = (V,E,w)$.
	Then any MST of $M_G$ is a spanning tree of $G$. (Hence there is a common MST for $G$ and $M_G$, denoted $Z$.)
\end{observation}
\begin{proof}
	Consider an MST $Z$ for $M_G$, and suppose for contradiction that $Z$ contains an edge $e$ outside $G$.
	Since $e$ belongs to $M_G \setminus G$, any path in $G$ between the endpoints of $e$ consists of at least two edges.
	Consider the (multi) graph obtained from $Z$ by replacing edge $e$ with a shortest path in $G$ between the endpoints of $e$.
	It is a spanning subgraph of $M_G$ of weight $w(Z)$, which contains at least $n+1$ edges (some of which may be multiple edges), and thus at least one cycle. By breaking cycles in this subgraph,
	we obtain a spanning tree of $M_G$ of weight strictly smaller than $w(Z)$, yielding a contradiction to the weight minimality of $Z$.
\end{proof}

\begin{lemma} \label{lightopt}
	Let $(M,\delta)$ be any metric space, $t \ge 1$ be some stretch parameter, and $H$ be the greedy $t$-spanner of $M$.
	For every $t$-spanner $H'$ of the metric space $M_H$ induced by $H$, we have $w(H) \le w(H')$.
\end{lemma}
\begin{proof}
	Let $H'$ be a $t$-spanner of $M_H$, and define $H''$ as the subgraph of $H$ obtained from $H'$ by replacing each edge $e$ of $H'$ with a shortest path in $H$ between
	the endpoints of $e$.
	Clearly, the distances in $H''$ are no greater than the respective distances in $H'$.
	Since $H''$ is a subgraph of $H$, it follows that $H''$ is a $t$-spanner for $H$.
	\lemmaref{obs:spanner_of_greedy} implies that $H'' = H$.
	Finally, noting that $w(H'') \le w(H')$, we have $w(H) = w(H'') \le w(H')$.
\end{proof}

The following lemma will be used  for proving the size optimality.
\begin{lemma}\label{lem:metric_spars}
	Let $(M,\delta)$ be any metric space, $t < 2$ be some stretch parameter, and $H$ be the greedy $t$-spanner of $M$.
	For every $t$-spanner $H'$ of the metric space $M_H$ induced by $H$, we have $|H| \le |H'|$.
\end{lemma}
\begin{proof}
	Denote by $w_H$ the weight function of $H$, i.e., for any edge $(u,v)\in H$, $w_H(u,v)=\delta(u,v)$,
	and for any path $P \in H$, $w_H(P) = \sum_{e \in P} w_H(e)$.
	Similarly, denote by $w_{H'}$ the weight function of $H'$, i.e., for any $(u,v)\in H'$, $w_{H'}(u,v)=\delta_{M_H}(u,v)$,
	and for any path $P \in H$, $w_{H'}(P) = \sum_{e \in P} w_{H'}(e)$.
	For every edge $e'\in H'$, let $P_{e'}$ be a shortest path between the endpoints of $e'$ in $H$; by definition, we have $w_{H'}(e')=w_H(P_{e'})$.
	We say that edge $e'\in H'$ \emph{covers} all edges of $P_{e'}$, and symmetrically, all edges of $P_{e'}$ are \emph{covered} by $e'$.
	(An edge $e' \in H \cap H'$ covers itself.)
	
	For each edge $e$ in $H \setminus H'$, let $\mathcal{Q}_{e}$ be a shortest path
	between the endpoints of $e$ in $H'$. Since $H'$ is a $t$-spanner for $M_H$, we have $w_{H'}(\mathcal{Q}_{e}) \le t\cdot w_H(e)$.
	Observe that the edges in $\cup_{e'\in\mathcal{Q}_{e}} P_{e'}$ form a  path $\Pi_e$ in $H$ between the endpoints of $e$.
	(It will be shown next that the path $\Pi_e$ is not simple.)
	We have
	$$w_{H}(\Pi_e) ~\le~ \sum_{e'\in\mathcal{Q}_{e}}w_{H}(P_{e'}) ~= \sum_{e'\in\mathcal{Q}_{e}}w_{H'}(e') ~=~ w_{H'}(\mathcal{Q}_{e}) ~\le~ t\cdot w_{H}(e)~.$$

	Next, we argue that the edge $e$ must belong to $\Pi_e$.
	Indeed, otherwise the edges of
	$\Pi_e$ contain a simple path in $H$ between the endpoints of $e$ of weight bounded by $t \cdot w_H(e)$,
	implying that the heaviest edge among the edges of this path and $e$ would not be added to the greedy $t$-spanner $H$.
	Consequently, at least one edge $e'$ in $\mathcal{Q}_{e}$ must cover $e$.
	
	We define an injection $f:H\rightarrow H'$ as follows. For each edge $e\in H \cap H'$, $f(e)$ is defined as $e$; in this case edge $e = f(e)$ covers itself.
	For each $e \in H \setminus H'$, $f(e)$ is defined to be an arbitrary edge of $\mathcal{Q}_{e}$ that covers $e$.
	To see that $f$ is injective, suppose for contradiction the existence of two distinct edges $e_{1}$ and $e_{2}$ in $H$
	and an edge $e' \in H'$	such that $f(e_{1})=f(e_{2})= e' \in H'$. It must hold that $e_1$ and $e_2$ are in $H \setminus H'$.
	Assume without loss of generality that $w(e_{1})\le w(e_{2})$.
	Since both $e_{1}$ and $e_{2}$ are covered by $e'$, it follows that $w_{H'}(e')\ge w_{H}(e_{1})+w_{H}(e_{2})\ge2\cdot w_{H}(e_{1})$.
	On the other hand, by the definition of $f$, the shortest path $\mathcal{Q}_{e_1}$ in $H'$ between the endpoints of $e_{1}$ contains the edge $e' = f(e_1)$.
	Hence the weight of a shortest path in $H'$ between the endpoints of $e_1$ is given by $w_{H'}(\mathcal{Q}_{e_1}) \ge w_{H'}(e') \ge2\cdot w_{H}(e_{1})>t\cdot w_{H}(e_{1})$, which contradicts	the fact that $H'$ is a $t$-spanner for $H$. It follows that $f$ is injective, from which we conclude that  $|H|\le|H'|$.
\end{proof}

Let $\mathcal{M}(n,\ddim)$ denote the family of $n$-point metric spaces with doubling dimension bounded by $\ddim$,
for any $n$ and $\ddim$.
The following observation shows that a small ``stretching'' of any metric space does not change the doubling dimension of the metric space by much.
\begin{observation}\label{lem:doubling_embedding}
	Let $H$ be a $t$-spanner of an arbitrary metric space $M \in \mathcal{M}(n,\ddim)$, for $t\le 2$.
	Then the metric space $M_H$ induced by $H$ belongs to $\mathcal{M}(n,2\ddim)$.
\end{observation}
\begin{proof}
	Clearly, any ball of radius $r$ in the ``stretched'' metric space $M_{H}$ is contained in the respective ball of the original metric space $M$.
	By definition, this ball can be covered by $2^{2\ddim}$ balls of radius $\frac{r}{4}$ in $M$,
	and so by $2^{2\ddim}$ balls of radius $t\cdot\frac{r}{4}\le\frac{r}{2}$ in the stretched metric space $M_H$.
\end{proof}	

The existential near-optimality result for doubling metrics is summarized in the following theorem.
\begin{theorem}[Greedy is near-optimal in doubling metrics]\label{theorem:doubling_light}
	For every metric $M \in\mathcal{M}(n,\ddim)$ and any stretch parameter $t < 2$, the greedy $t$-spanner $H$ of $M$
	has at most $OPT^{sparse}_t(\mathcal{M}(n,2\ddim))$ edges and lightness at most $OPT^{light}_t(\mathcal{M}(n,2\ddim))$.
\end{theorem}
\begin{proof}
	Let $M$ be an arbitrary metric space in $\mathcal{M}(n,\ddim)$, let $H=\left(V,E\right)$ be
	the greedy $t$-spanner for $M$,	and let $M_H$ be the metric space induced by $H$.	
	By \observationref{lem:doubling_embedding}, $M_H \in \mathcal{M}(n,2\ddim)$.
	Hence, there exist $t$-spanners $\mathcal{H}^{sparse}$ and $\mathcal{H}^{light}$ for $M_H$ with at most $OPT^{sparse}_t(\mathcal{M}(n,2\ddim))$
	edges and lightness at most $OPT^{light}_t(\mathcal{M}(n,2\ddim))$, respectively.
	\lemmaref{lem:metric_spars} implies that  $|H|\le |\mathcal H^{sparse}|$, from which the size bound on $H$ immediately follows.
	As for the lightness bound on $H$, note that $\mathcal{H}^{light}$ is computed on top of $M_H$ rather than the original metric space $M$.
	Nevertheless, \observationref{fct:greedy contains MST} and \observationref{mst-metric} imply that $M$ and $M_H$ have the same MST $Z$.
	Since the lightness of $\mathcal{H}^{light}$ is at most $OPT^{light}_t(\mathcal{M}(n,2\ddim))$, it follows that
	$$OPT^{light}_t(\mathcal{M}(n,2\ddim)) ~\ge~ \Psi(\mathcal{H}^{light}) ~=~ \frac{w(\mathcal{H}^{light})}{w(MST(M_H))} ~=~ \frac{w(\mathcal{H}^{light})}{w(Z)},$$
	hence $w(\mathcal{H}^{light}) \le OPT^{light}_t(\mathcal{M}(n,2\ddim)) \cdot w(Z)$.
	By \lemmaref{lightopt}, we have $w(H) \le w(\mathcal {H}^{light})$, hence the lightness of $H$ satisfies
	\begin{eqnarray*} \Psi(H)  = \frac{w(H)}{w(MST(M))} ~=~ \frac{w(H)}{w(Z)} ~\le~ \frac{w(\mathcal{H}^{light})}{w(Z)}
		~\le~ OPT^{light}_t(\mathcal{M}(n,2\ddim))~.
	\end{eqnarray*}
\end{proof}

By \theoremref{thm:smid09}, the greedy $(1+\eps)$-spanner for $n$-point doubling metrics
has $O(n)$ edges and lightness $O(\log n)$,
where the $O$-notation hides a multiplicative term of $(1/\epsilon)^{O(\ddim)}$.
Applying \theoremref{theorem:doubling_light} in conjunction with  \theoremref{thm:Got15}, we reduce the lightness bound of the greedy $(1+\eps)$-spanner to constant.
\begin{corollary}\label{cor:Greedy_G}
	For every metric space $(M,\delta)$  in $\mathcal{M}(n,\ddim)$ and any parameter $0<\epsilon<\frac{1}{2}$,  the greedy $(1+\epsilon)$-spanner has
	$n(1/\epsilon)^{O(\ddim)}$  edges and lightness $(\ddim/\epsilon)^{O(\ddim)}$.
\end{corollary}
\noindent {\bf Remark.}
\corollaryref{cor:Greedy_G} shows that the greedy $(1+\eps)$-spanner in doubling metrics achieves optimal bounds on the size and the lightness, \emph{disregarding dependencies on $\eps$ and the doubling dimension}.
However, improving these dependencies is a fundamental challenge of practical importance.
By \theoremref{theorem:doubling_light}, any improvement whatsoever in the dependencies on $\eps$ and the doubling dimension
on either the size or the lightness of \emph{any spanner construction} for doubling metrics -- would trigger a similar improvement to the greedy spanner.
Note that the recent result of Borradaile \etal ~\cite{BLW19} shows that the greedy $(1+\eps)$-spanner has lightness  $(1/\epsilon)^{O(\ddim)}$

\section{The Approximate-Greedy Spanner in Doubling Metrics is Light}\label{sec:Fast_alg}
\corollaryref{cor:Greedy_G} shows that the greedy $(1+\eps)$-spanner in doubling metrics achieves near-optimal bounds on the size and the lightness.
Nevertheless, this spanner has two major disadvantages.
First, as mentioned in the introduction, there exist metric spaces with doubling dimension 1 for which its degree may be unbounded.
(This is in contrast to $d$-dimensional Euclidean metrics, where the greedy $(1+\eps)$-spanner has degree $(1/\eps)^{O(d)}$.)
Second, it cannot be constructed within sub-quadratic time in doubling metrics due to a lower bound of \cite{HM06}.
In fact, even in $d$-dimensional Euclidean metrics, the state-of-the-art implementation of the greedy $(1+\eps)$-spanner requires time $(1/\eps)^{O(d)} (n^2 \log n)$ \cite{BCFMS10}.

\sloppy Building on \cite{DHN93,DN97}, Gudmundsson et al.\ \cite{GLN02} devised a much faster algorithm that follows the greedy approach,
hereafter Algorithm \texttt{Approximate-Greedy}.
The runtime of this algorithm is $(1/\eps)^{O(d)}(n \log n)$, yet the degree and lightness of the approximate-greedy spanner produced by the algorithm are both bounded by $(1/\eps)^{O(d)}$, just as with the greedy spanner for Euclidean metrics.
The runtime analysis of Algorithm \texttt{Approximate-Greedy} \cite{GLN02} does not exploit any properties of Euclidean geometry.
Specifically, it relies on the triangle inequality, which applies to arbitrary metric spaces, and on standard packing arguments (cf.\ \lemmaref{lem:doubling_packing}), which apply to arbitrary doubling metrics.
Therefore, the runtime of Algorithm \texttt{Approximate-Greedy} remains $(1/\eps)^{O(d)}(n \log n)$ in arbitrary doubling metrics.
Moreover, the degree bound of  $(1/\eps)^{O(d)}$ applies to arbitrary doubling metrics as well.
(We refer to Chapter 15 in \cite{NS07} for an excellent description of this algorithm and its analysis.)

In this section we show that the approximate-greedy spanner of \cite{GLN02} has constant lightness in arbitrary doubling metrics.
Consequently, Algorithm  \texttt{Approximate-Greedy} provides an $O(n\log n)$-time construction of $(1+\eps)$-spanners in doubling metrics with lightness and degree both bounded by constants.

\subsection{A Rough Sketch of Algorithm \texttt{Approximate-Greedy}}
In this section we provide a  very rough sketch of Algorithm \texttt{Approximate-Greedy}, aiming to highlight the high-level ideas behind it.
This outline is not required for the analysis that is given in \sectionref{bounding}; it is provided here for clarity and completeness.

In metric spaces, the greedy algorithm sorts the ${n \choose 2}$ interpoint distances and examines the edges by non-decreasing order of weight.
For each edge that is examined for inclusion in the spanner, the distance between its endpoints in the current spanner is computed.
This is expensive for two reasons: (1) The number of examined interpoint distances is quadratic in $n$.
(2) Computing the \emph{exact} spanner distance between two points is   costly.

Suppose we aim for a stretch of $t = 1+\eps$, and let $t'$ be an appropriate parameter satisfying $t' = 1 + O(\eps) < t$.
(Refer to \cite{GLN02,NS07} for the exact constant hiding in the $O$-notation of $O(\eps)$.)
Instead of examining all ${n \choose 2}$ interpoint distances, Algorithm \texttt{Approximate-Greedy} computes a
bounded degree $\sqrt{t/t'}$-spanner $G' = (M,E',\delta)$ for the input metric space $(M,\delta)$,
and simulates the greedy algorithm with stretch parameter $\sqrt{t \cdot t'}$ only on the edges of $G'$.
The output of the algorithm is a $\sqrt{t \cdot t'}$-spanner $G = (M,E,\delta)$ for $G'$,
which is a $t$-spanner for the original metric space $(M,\delta)$ by the ``transitivity'' of spanners.
A spanner $G'$ of degree $(1/\eps)^{O(\ddim)}$ can be constructed in $(1/\eps)^{O(\ddim)} (n \log n)$ time via \theoremref{thm:doubling_degree}.
Since the output $t$-spanner $G$ for $(M,\delta)$ is a subgraph of $G'$, its degree will be at most $(1/\eps)^{O(\ddim)}$.

The greedy simulation is applied only on the edges of $G'$ that are sufficiently ``heavy''.
Formally, let $D$ denote the maximum weight of any edge of the bounded degree spanner $G'$, and let $E_0$ be the set of \emph{light edges} in $E'$, namely, of weight at most $D/n$.
As $|E_0| \le |E'| = O(n)$, we have $w(E_0) = O(D) = O(MST(M))$.
All light edges are taken to the output spanner $G$, and the greedy simulation is applied only on the edges of $E' \setminus E_0$.
(So the output spanner $G$ will contain \emph{all} edges of $E_0$ and \emph{some} edges of $E' \setminus E_0$.)

As mentioned, computing the \emph{exact} distance between two points is costly;
using Dijkstra's algorithm, it requires $O(n \log n)$ time (see, e.g., Section 2.5 and Corollary 2.5.10 in \cite{NS07}).
Since $G'$ has $O(n)$ edges, the overall runtime will be $O(n^2 \log n)$.
To speed up the computation time, Algorithm \texttt{Approximate-Greedy} does not compute the exact distance between two points, but rather an approximation of that distance.
This is achieved by maintaining a much simpler and coarser \emph{cluster graph} that approximates the original distances, on which the distance queries are performed.
More specifically, the algorithm partitions the edge set $E' \setminus E_0$ into $\log_\mu n$ buckets,  for an appropriate parameter $1 < \mu = O(\log n)$,
such that edge weights within each bucket differ by at most a factor of $\mu$.
Then it examines the edges of $E' \setminus E_0$ by going from one bucket to the next, examining edges by non-decreasing order of weight.
Whenever all edges of some bucket have been examined, the cluster graph is updated according to the new edges that were added to the spanner.
The idea is to periodically make the cluster graph simpler and coarser, so that the shortest path computations made on it will be fast.
The bottom-line is that one does not simulate the greedy algorithm (with stretch parameter $\sqrt{t \cdot t'}$) on the edge set $E' \setminus E_0$, but rather an \emph{approximate version} of it.

\subsection{Bounding the Lightness of the  \texttt{Approximate-Greedy} Spanner} \label{bounding}
As mentioned, the runtime of Algorithm \texttt{Approximate-Greedy} is $(1/\eps)^{O(\ddim)}(n \log n)$  in arbitrary doubling metrics.
In what follows let $G = (M,E,\delta)$ be the $t$-spanner for $(M,\delta)$  returned by Algorithm \texttt{Approximate-Greedy}.
Since $G$ is a subgraph of the bounded degree spanner $G' = (M,E',\delta)$, its degree is $(1/\eps)^{O(\ddim)}$.

It remains to bound the lightness of $G$.
The lightness argument of \cite{GLN02}, which relies on previous works \cite{DHN93,DN97},
is based on rather deep properties from Euclidean geometry, most notably the \emph{leapfrog property}.
In particular, this argument does not apply to arbitrary doubling metrics.

Instead, we employ the following lemma, which lies at the heart of the lightness analysis of \cite{GLN02}.
While this lemma applies to arbitrary doubling metrics, the way it was used in \cite{GLN02} does not extend to arbitrary doubling metrics.
Specifically, it was used in \cite{GLN02} to show that the edge set $E \setminus E_0$ satisfies the leapfrog property.
(Recall that $E_0$ is the set of light edges in $G' = (M,E',\delta)$, all of which are taken to the approximate-greedy spanner $G = (M,E,\delta)$.)
In Euclidean metrics, it has been proved \cite{DHN93,NS07} that any edge set satisfying the leapfrog property has constant lightness, but this proof does not carry over to arbitrary doubling metrics.
\begin{lemma}[Lemma 17 in \cite{GLN02}]\label{lem:GLN_second_path}
	Let $e=(u,v) \in E \setminus E_0$. The weight of the second shortest path between $u$ and $v$ in the approximate-greedy spanner $G$ is greater than $t' \cdot w(e)$.
	(If there are multiple shortest paths between $u$ and $v$, then the weight of the second shortest path equals the weight of the shortest path.)
\end{lemma}
\noindent
{\bf Remark.} The parameter $t'$ in the statement of this lemma depends on the stretch parameters of the spanners $G'$ and $G$
that are constructed by Algorithm \texttt{Approximate-Greedy}.
Specifically, recall that the output spanner $G$ is a $\sqrt{t \cdot t'}$-spanner for $G'$,
which is, in turn, a $\sqrt{t/t'}$-spanner for the input metric space $M$.

We will use the following observation, due to \cite{Smid07}. We include a proof for completeness.
\begin{observation} [Lemma 1.7 in \cite{Smid07}] \label{mstsimple2}
	Let $H$ be an arbitrary weighted graph, and let $t$ be any stretch parameter.
	For any $t$-spanner $H'$ of $H$, $w(MST(H')) \le t \cdot w(MST(H))$.
\end{observation}
\begin{proof}
	Consider an MST $Z$ for $H$.
	Replace each edge of $Z$ by a $t$-spanner path in $H'$ between the endpoints of that edge, and then break cycles.
	The resulting structure $Z'$ is a spanning tree of $H'$, hence $w(MST(H')) \le w(Z')$, and we have
	$w(MST(H')) \le w(Z') \le t \cdot w(Z) = t \cdot w(MST(H))$.
\end{proof}

The following lemma bounds the lightness of $G$.
Its proof is based on the somewhat surprising observation that the lightness of the $t$-spanner $G$ produced by Algorithm \texttt{Approximate-Greedy} is existentially near-optimal with respect to stretch parameter $t' < t$ (rather than $t$).
We remark that $G$ is not a greedy spanner, but rather an approximate-greedy spanner,
and it is inherently different than the greedy $t$-spanner and the greedy $t'$-spanner.
In particular, its weight may be larger than the weights of both these greedy spanners.
Nevertheless, our existential near-optimality argument suffices to derive the required lightness bound.
\begin{lemma}
	The lightness of $G$ is $(\frac{\ddim}{t'-1})^{O(\ddim)}$.
\end{lemma}

\begin{proof}					
	Recall that $G = (M,E,\delta)$ is a $t$-spanner for $M$, where $t = 1+\eps, \eps < 1$,  and let $M_G$ be the metric space induced by $G$.
	By \observationref{lem:doubling_embedding}, the doubling dimension of $M_G$ is bounded by $2 \ddim$.
	Let $H'$ be a $t'$-spanner of $M_G$ with lightness $OPT^{light}_{t'}(\mathcal{M}(n,2\ddim))$, where $t' = 1+O(\eps) < t$
	is the parameter appearing in the statement of \lemmaref{lem:GLN_second_path}, which is optimized as part of Algorithm  \texttt{Approximate-Greedy}.
	As in the proof of \lemmaref{lightopt}, we transform $H'$ into a $t'$-spanner $H''$ of $G$ of weight at most $w(H')$.
	By \observationref{mst-metric} and \observationref{mstsimple2}, the MST weights for all graphs $M,G,M_G,H'$ and $H''$ are the same, up to a factor of $t \cdot t' = O(1)$.

	We argue that every edge $e \in E \setminus E_0$ belongs to $H''$.
	Suppose for contradiction that there is an edge $e \in E \setminus E_0$ that does not belong to $H''$.
	Let $P$ be a shortest path between the endpoints of $e$ in $H''$. Since $H''$ is a $t'$-spanner of $G$, we have $w(P) \le t' \cdot w(e)$.
	Note that this path is contained in $G$. Since $e \in G$ and $M$ is a metric space, the weight of the second shortest path between the endpoints of $e$ in $G$ is at most
	$w(P) \le t' \cdot w(e)$.
	On the other hand, By \lemmaref{lem:GLN_second_path}, the weight of this path is greater than $t'\cdot w(e)$, a contradiction.
	It follows that
	\begin{eqnarray*}
		w(G) & = & w(E\setminus E_{0}) +w(E_{0}) ~\le~ w(H'')+w(E_{0})\\
		& \le & w(H')+w(E_{0}) ~=~ \left(\frac{\ddim}{t'-1}\right)^{O(\ddim)} \cdot w(MST(M))~.
	\end{eqnarray*}
\end{proof}

Setting $t = 1+\eps$ and $t' = 1+c \cdot \eps$ (for an appropriate constant $c$; see \cite{GLN02,NS07}), we conclude:
\begin{theorem}  \label{finish} \sloppy
	For any  metric  space $(M,\delta)$ in $\mathcal{M}(n,\ddim)$ and parameter $0<\epsilon<\frac{1}{2}$, Algorithm \texttt{Approximate-Greedy}
	returns a $(1+\eps)$-spanner with lightness $\left(\frac{\ddim}{\eps}\right)^{O(\ddim)}$ and degree $(1/\eps)^{O(\ddim)}$.
	The runtime of Algorithm \texttt{Approximate-Greedy} is $(1/\eps)^{O(\ddim)}(n \log n)$.
\end{theorem}
\noindent{\bf Remark.}	\theoremref{finish} should be compared to \theoremref{thm:Got15} due to \cite{Got15}.
Both constructions achieve the same lightness bound, but the degree and number of edges in the spanner construction of \cite{Got15} are unbounded.
Moreover, the runtime of the construction of \cite{Got15} is $(\ddim/\epsilon)^{O(\ddim)} (n\log^{2}n)$, whereas that
of \theoremref{finish} is $(1/\eps)^{O(\ddim)}(n \log n)$.
By combining the light spanner $H_1$ of \cite{Got15} with a bounded degree spanner $H_2$, one can obtain a spanner with constant degree and   lightness.
Specifically, such a spanner $\mathcal{H}$ is obtained by replacing each edge of $H_1$ with a shortest (or approximately shortest) path in $H_2$ between the endpoints of that edge.
The lightness of the resulting spanner $\mathcal{H}$ will not exceed that of $H_1$ by much, whereas the degree bound  will follow from that of $H_2$.
There is a major problem with this approach: The runtime needed for computing spanner $\mathcal{H}$ may be very high.
Indeed, although there are efficient ways to estimate the weight of an approximately shortest path in $H_2$ between two points,
we must \emph{compute} the corresponding path in $H_2$.
In particular, to achieve the degree bound of $H_2$, one may not use edges outside $H_2$.
Moreover, even regardless of this computation time, such a path may contain many edges that already belong to the gradually growing spanner $\mathcal{H}$.
Deciding which edges of this path should be added to $\mathcal{H}$ may be very costly by itself.\\
By the recent result of Borradaile \etal ~\cite{BLW19}, the lightness of the spanner construction provided by   \theoremref{finish} is reduced to $(1/\epsilon)^{O(\ddim)}$.

\section{Acknowledgements}  We are grateful to Michael Elkin and Ofer Neiman for fruitful discussions.
The second-named author thanks L\'{a}szl\'{o} Babai for comments that helped improving the presentation of the paper.

\bibliographystyle{alpha}
\bibliography{bibShay,ENS14,ENS14V2}

\end{document}